\documentclass{article}
\usepackage[utf8]{inputenc}
\usepackage{amsthm}
\usepackage{amsmath}
\usepackage{amsfonts}
\usepackage{mathrsfs} 
\usepackage{mathtools} 
\usepackage{nccmath}
\usepackage{bussproofs}
\usepackage{enumerate}
\usepackage{color}

\newtheorem{theorem}{Theorem}[section]

\newtheorem{claim}{Claim}[theorem]
\newtheorem{prop}[theorem]{Proposition}
\newtheorem{lemma}[theorem]{Lemma}
\newtheorem{definition}{Definition}[section]
\newtheorem{notation}{Notation}[section]

\theoremstyle{remark}

\title{On principal types and well-foundedness of the cummulativity relation in ECC}
\author{
Eitetsu Ken\footnote{
Eitetsu Ken,the University of Tokyo, Graduate School of Mathematical Sciences,Japan, \texttt{yeongcheol-kwon@g.ecc.u-tokyo.ac.jp}}, 
Masaki Natori\footnote{
Masaki Natori,the University of Tokyo, Graduate School of Mathematical Sciences,Japan, \texttt{natori-masaki616@g.ecc.u-tokyo.ac.jp}},
Kenji Tojo\footnote{
Kenji Tojo, the University of Tokyo, Department of Information Science, Faculty of Science, Japan,
\texttt{kenjitojo@is.s.u-tokyo.ac.jp}},
Kazuki Watanabe\footnote{Kazuki Watanabe, the Graduate University for Advanced Studies, Japan\newline
\texttt{kazukiwatanabe@nii.ac.jp}}}

\begin{document}
\maketitle
\begin{abstract}
    When we investigate a type system, it is helpful if we can establish the well-foundedness of types or terms with respect to a certain hierarchy, and the Extended Calculus of Constructions (called $ECC$, defined and studied comprehensively in \cite{Luo}) is no exception. However, under a very natural hierarchy relation 
   (called the cumulativity relation in \cite{Luo}),
     the well-foundedness of the hierarchy
      does not hold generally.
    In this article,we show that the cumulativity relation is well-founded if it is restricted to one of the following two natural families of terms:
   
    \begin{enumerate}
     \item types in a valid context
     \item terms having normal forms
    \end{enumerate}
    Also, we give an independent proof of the existence of principal types in $ECC$ since it is used in the proof of well-foundedness of cumulativity relation in a valid context although it is often proved by utilizing the well-foundedness of the hierarchy, which would make our argument circular if adopted.
\end{abstract}

\section{Introduction}
 \quad When we investigate a type system, it is helpful to prove the well-foundedness of a certain hierarchy on types or terms, and this strategy is taken in \cite{Luo}, too, in order to study the Extended Calculus of Constructions (called $ECC$).
 However, for a very natural hierarchy relation (namely, the one which is called the cumulativity relation in \cite{Luo}), the well-foundedness does not hold globally.
 This fact was already noticed by Luo.Z in 2001 (unpublished), and at the same time, it was clarified by him that Lemma 3.5 and Corollary 3.8 in \cite{Luo} were incorrect. However, the two statements are used only for establishing the existence of principal types of objects in a valid context, and he was made aware of the fact that it can actually be proved without utilizing the two incorrect statements.
 Later, in 2019, the authors rediscovered counterexamples for the two statements (which were exactly the same as those Luo. Z had been made aware of in 2001) without the knowledge of his unpublished works, and found some variations of well-foundedness of types or terms which are indeed true, and proved them on our own. \\
 \quad In this article, first we will give counterexamples to Lemma 3.5 and Corollary 3.8 in \cite{Luo} . Then, we prove that the cumulativity relation is indeed well-founded if it is restricted to either one of the following two natural families of terms:
 \begin{enumerate}
  \item types in a valid context
  \item terms having normal forms
 \end{enumerate}
 
 Furthermore, following Luo's unpublished work stated above, we give a precise proof of the existence of principal types in $ECC$ since it is used in the proof of our first version of well-foundedness (it plays an important role in proving quasi-normalisation) although it is often proved by utilizing the well-foundedness of the cumulativity relation, which would make our argument circular if adopted.

\section{Acknowledgement}
\quad The authors are really grateful to Professor Luo. Z for giving us valuable comments and sharing his unpublished works. We also would like to thank Fukuda. Y for his helpful comments to improve this article.

\section{Proofs of well-foundedness of $\preceq$ with additional condition}
\quad Following the notation of \cite{Luo}, we assume $\epsilon$ as an environment.\\
\null \quad In order to show the well-foundedness of the cumulativity relation, the following statement would play a key role if it could be established (since the claim would reduce the well-foundedness of $\preceq$ to that of $\preceq_i$);

\begin{claim}
 Let $A_1, A_2, B_1,B_2$ be terms, and $i \in \omega$. Then,
 \begin{align*}
 A_1 \preceq A_2 \prec_i B_1 \preceq B_2 \Rightarrow A_1 \preceq_i A_2 \& B_1 \preceq_i B_2. 
 \end{align*}
\end{claim}

However, as for $ECC$, the claim does not hold; 
\begin{prop}
Let
  \begin{align*}
  C &:= \Sigma x : (\Sigma y: Prop. Prop).Prop = (Prop \times Prop) \times Prop \\
  A &:= \Sigma x : (\Sigma y: Prop. Type_0).Prop = (Prop \times Type_0) \times Prop \\
  B &:= \Sigma x : (\Sigma y: Prop. Type_0).Type_0 = (Prop \times Type_0) \times Type_0 \\
  \end{align*}
 Then, $C \preceq A \prec_1 B$, but $ C \preceq_1 A $ does not hold.
\end{prop}

 \begin{proof}
  By definition, $C \preceq A \preceq_1 B$ holds. Also, by the Church-Rosser theorem, $A$ and $B$ are not computationally equal. \\
  \quad Suppose $ C \preceq_1 A $ holds. By the Church-Rosser theorem, 
  it turns out that the case (c) of the definition of $\preceq_1$ is applied, 
  hence, using the Church-Rosser theorem again, we get
  \begin{align*}
  \Sigma y: Prop. Prop \preceq_0 \Sigma y: Prop. Type_0 .
  \end{align*} 
  
  However, since the both of them are in their normal forms, this is impossible because of the Church-Rosser theorem and the definition of 
  $\preceq_0$.
 \end{proof}

 Does the well-foundedness hold, then? Considering arbitrary terms of $ECC$, unfortunately it fails:
  
\begin{prop}
 There exists an infinite decreasing sequence of terms $A_0 \succ A_1 \succ \ldots$.
\end{prop}
 
\begin{proof}
 Let a term $\alpha$ be
 \begin{align*}
  \alpha := \lambda y. (\Sigma x : Type_0. yy) 
 \end{align*}
 Then,
 \begin{align*}
  \alpha \alpha \simeq \Sigma x : Type_0. (\alpha \alpha)
 \end{align*}
 holds.
 Using this and the Church-Rosser theorem again and again,
  we can get the following infinite decreasing sequence.
 \begin{align*}
  \alpha \alpha 
  &\simeq \Sigma x : Type_0. (\alpha \alpha)\\
  &\succ \Sigma x : Prop. (\alpha \alpha)\\
  &\simeq \Sigma x : Prop. (\Sigma x : Type_0. (\alpha \alpha))\\
  &\succ \Sigma x : Prop. (\Sigma x : Prop. (\alpha \alpha)) \ldots
 \end{align*}
\end{proof}

However, the term $\alpha$ above is clearly not ``well-typed." 
 Hence, imposing some regularity on the terms, we can get the desired well-foundedness.
 We will show two specific cases in which the well-foundedness actually hold.\\
 \null \quad The first one restricts the terms to $\epsilon$-types:

%

\begin{theorem} (well-foundedness) \\
There is no infinite decreasing sequence of $\epsilon$-types $A_0 \succ A_1\succ\dots A_n\succ\dots$.
\end{theorem}

\begin{proof}
We can assume every $A_n (n\in \omega)$ is a quasi-normal form by \cite{Luo}. We prove this lemma by induction on the structure of quasi-normal form $A_0$.\\
Case(1): If $A_0\equiv$ $U$(universe), it is obvious that the statement holds. \\
Case(2): Suppose $A_0$ is a base term and if $A_0\succeq A_1$ , then $A_0\simeq A_1$ by the property of the cumulativity relation.\\
Case(3): Assume $A_0\equiv$ $\Pi x:A_0.B_0$ and $A_0 \succ A_1\succ\dots A_n\succ\dots$ holds. By the property of the cumulativity relation, every $A_n$ is a form of $\Pi x:B_n. C_n$ such that $C_0\succ C_1\succ \dots \succ C_n \succ\dots$ , which contradicts the induction hypothesis.\\
Case(4): Assume $A_0\equiv$ $\Sigma x:A_0.B_0$ and $A_0 \succ A_1\succ\dots A_n\succ\dots$ holds. By the property of the cumulativity relation, every $A_n$ is a form of $\Sigma x:B_n. C_n$. By the assumption and the property of the cumulativity relation, $B_0 \succ B_1\succ \dots B_n\succ \dots$ or $C_0 \succ C_1\succ \dots C_n\succ \dots$ hold, which contradicts the induction hypothesis.

\end{proof}

The next one restricts the terms to just normalisable ones (hence, the result depends only on the syntax of terms and the definition of the type cumulativity. Especially, the proof system does not matter):

 \begin{theorem}\label{nfwf}
  There is no infinite decreasing sequence of terms $\{ A_n\}_{n \in \omega}$ such that each $A_n$ has a normal form. 
 \end{theorem} 
 
 To show this, we will define some notions and establish their basic properties.
 
 \begin{definition}
  Let $\mathcal{T}$ be the set of all terms which have normal forms.\\
  Then, the following stratification of $\mathcal{T}$ is defined:
  \begin{align*}
    \Pi_0 &:= \Sigma_0 := \mathcal{T} \setminus \overline{\{ A \in \mathcal{T} \mid \mbox{$A$ has $\Pi$ or $\Sigma$ as its outermost symbol.} \} }\\
    \Pi_{n+1}&:= \overline{\{ \Pi x: A_1. A_2 \in \mathcal{T} \mid \exists k,l \leq n\ (A_1 \in \Pi_k \cup \Sigma_k)\&(A_2 \in \Pi_l \cup \Sigma_l)\&(k+l=n) \}}\\
    \Sigma_{n+1}&:= \overline{\{ \Sigma x: A_1. A_2 \in \mathcal{T} \mid \exists k,l \leq n\ (A_1 \in \Pi_k \cup \Sigma_k)\&(A_2 \in \Pi_l \cup \Sigma_l)\&(k+l=n) \}}    
   \end{align*}
      (where $\overline{S}$ denotes the closure of $S$ with respect to conversions.)
 \end{definition}
 
 The following properties are verified:
 \begin{lemma}\label{basics}
   \begin{enumerate}
    \item \label{sub} $\forall n \in \omega.\ \Pi_n, \Sigma_n \subset \mathcal{T}$.
    \item \label{cong} Let $n \in \omega$, and $M,N \in \mathcal{T}$. Then, the following holds: 
    \begin{align*}
    N \simeq M, M \in \Pi_n \Rightarrow N \in \Pi_n
    \end{align*}
    This also holds for $\Sigma_n$.
    \item \label{union} $\mathcal{T} = \bigcup_{n \in \omega} (\Pi_n \cup \Sigma_n)$.
    \item \label{psdisj} $\forall n \geq1.\ \Pi_n \cap \Sigma_n = \emptyset $.
    \item \label{ppdisj} Let $i,j \in \omega$. Then, the following holds; 
    \begin{align*}
    M \in \Pi_i, N \in \Pi_j , M \simeq N \Rightarrow i=j
    \end{align*}
    \item \label{ssdisj} Let $i,j \in \omega$. Then, the following holds; 
    \begin{align*}
    M \in \Sigma_i, N \in \Sigma_j , M \simeq N \Rightarrow i=j
    \end{align*}
   \end{enumerate}
 \end{lemma}
 
 \begin{proof}
 (\ref{sub}), (\ref{cong}) : clear.\\
 \quad (\ref{union}) : By (\ref{sub}), it suffices to show that every $A \in \mathcal{T}$ is in $\Pi_n$ or $\Sigma_n$ for some $n \in \omega$. This can be proved by induction on the construction of the normal form of $A$. \\
 \null \quad When $A$ is convertible to a constant or variable, $A \in \Pi_0=\Sigma_0$ by Church-Rosser Theorem.\\
 \null \quad Similarly, when the normal form of $A$ is in the form of $(\lambda x:M. N)$ or $\langle M,N\rangle_B$ (where $M$, $N$ and $B$ are in their normal forms), $A \in \Pi_0=\Sigma_0$ by Church-Rosser Theorem.\\
 \null \quad Similar arguments can be applied to the case when the normal form of $A$ is in the form of $MN$ or $\pi_i(M)$ (where $M$ and $N$ are in their normal forms, $i =1,2$), since they do not form redexes.\\
 \null \quad When the normal form of $A$ is in the form of $\Pi x:M.N$ or $\Sigma x:M.N $ (where $M$, $N$ are in their normal forms), then, by the induction hypothesis, $M \in \Pi_k \cup \Sigma_k$ and $N \in \Pi_l \cup \Sigma_l$ for some $k,l \in \omega$. Hence, $A \in \Pi_{k+l+1}\cup \Sigma_{k+l+1}$.\\
 (\ref{psdisj}) : This easily follows from Church-Rosser Theorem.\\
 (\ref{ppdisj}) : This can be proved by induction on $i$.\\
 \null \quad When $i=0$, the statement follows immediately from the definitions of $\Pi_n$'s. \\
 \null \quad When $i>0$, suppose $M \simeq N$ and $M \in \Pi_i$, $N \in \Pi_j$. Since the case $j=0$ has already been dealt with, we may assume that $j>0$. Then, $M$ and $N$ can be written as:
  \begin{align*}
   M \simeq \Pi x:M_1. M_2,\ N \simeq \Pi x: N_1. N_2
  \end{align*}
  (where $M_1 \in \Pi_{k_1}$, $M_2 \in \Pi_{k_2}$ for some $k_1,k_2$ such that $k_1+k_2=i-1$, and 
  $N_1 \in \Pi_{l_1}$, $N_2 \in \Pi_{l_2}$ for some $l_1,l_2$ such that $l_1+l_2=j-1$)\\
  \quad Since $M \simeq N$, using Church-Rosser theorem, we get $M_1 \simeq N_1$ and $M_2 \simeq N_2$. Hence, $k_1=l_1$ and $k_2 =l_2$ follow from the induction hypothesis. Therefore, $i=k_1+k_2+1=l_1+l_2+1=j$.\\
  (\ref{ssdisj}): This can be proved analogously to (\ref{ppdisj}).
 \end{proof} 
  
 Now, we are ready to prove Theorem \ref{nfwf}.
 
 \begin{proof}[Proof of Theorem \ref{nfwf}]
   We will define a map
  \begin{align*}
   \varphi \colon \mathcal{T} \rightarrow \omega
  \end{align*}
  such that $A \prec B \Rightarrow \varphi(A) < \varphi(B)$. Then, the well-foundedness of $\omega$ will imply that of $\mathcal{T}$. \\
  \null \quad For $T \in \mathcal{T}$, $\varphi(T) \in \omega$ is defined as follows:
   \begin{enumerate}
    \item Consider the case in which $T \in \Pi_0 =\Sigma_0$.
     \begin{enumerate}
      \item When $T \simeq Prop$, $\varphi(T):=2$.
      \item When $T \simeq Type_j$ (where $j \in \omega$), $\varphi(T):=3+j$.
      \item Otherwise, $\varphi(T):=1$.
     \end{enumerate}
     $\varphi$ is well-defined on $\Pi_0=\Sigma_0$ by Church-Rosser Theorem, and clearly $\varphi$ is invariant with respect to conversions.
    \item Suppose $\varphi$ has been defined on $\bigcup_{j \leq n}(\Pi_j \cup \Sigma_j)$, and is invariant with respect to conversions. Consider the case in which $T \in \Pi_{n+1}$.
     $T$ can be written as $T \simeq \Pi x :A. B$ (where $A \in \Pi_k \cup \Sigma_k$ and $B \in \Pi_l \cup \Sigma_l$ for some $k,l \in \omega$ such that $k+l=n$). Using this, let $\varphi(T) := \varphi(A) \varphi(B)$ (this value does not depend on the choice of $A$ and $B$. Suppose $T \simeq \Pi x :A^{\prime}. B^{\prime}$. Then, we get $A \simeq A^{\prime}$ and $B \simeq B^{\prime}$ by Church-Rosser Theorem. Since $\varphi$ is invariant with respect to conversions on $\bigcup_{j \leq n}(\Pi_j \cup \Sigma_j)$, $\varphi(A)=\varphi(A^{\prime})$ and $\varphi(B)=\varphi(B^{\prime})$ hold).\\
     \null \quad The case in which $T \in \Sigma_{n+1}$ is to be dealt with similarly.\\
     \null \quad The argument that showed the well-definedness of $\varphi$ on $\Pi_{n+1} \cup \Sigma_{n+1}$ also show the invariance of $\varphi$ with respect to conversions. 
     Notice also that $\varphi$ is a function on $\bigcup_{j \leq n+1}(\Pi_j \cup \Sigma_j)$ since $\Pi_i$'s, $\Sigma_i$'s ($i \geq 1$) and $\Pi_0=\Sigma_0$ are mutually disjoint by Lemma \ref{basics} (\ref{psdisj}), (\ref{ppdisj}), (\ref{ssdisj}).  
   \end{enumerate}
   
   Now, $\varphi \colon \mathcal{T} \rightarrow \omega$ has been defined. Let us show that $M \prec N$ implies $\varphi(M) < \varphi(N)$ for $M,N \in \mathcal{T}$. It suffices to show the following claim by induction on $i$: \\
   \centerline{Let $i \in \omega$ and $M,N \in \mathcal{T}$. Then, $M \prec_i N$ implies $\varphi(M) < \varphi(N)$.}\\
   \quad When $i=0$, $M \prec_i N$ means $M$ and $N$ are both universes and $M \prec N$. 
   Hence, $\varphi(M) < \varphi(N)$ by the definition of $\varphi$.\\
   \quad When $i>0$, $M \prec_i N$ means ``$M \prec_{i-1} N$, or $M \simeq Q x:M_1.M_2 \& N \simeq Q x:N_1.N_2$ (where $M_j \preceq_{i-1} N_j$ for each $j$, $M_j \prec_{i-1} N_j$ for some $j$, and $Q=\Pi$ or $\Sigma$)."
   In the former case, the claim follows from the induction hypothesis. 
   In the latter case, $M_1,M_2,N_1,N_2 \in \mathcal{T}$ since $M,N \in \mathcal{T}$ and Church-Rosser Theorem. Since $M_1 \prec_{i-1} N_1$ or $M_2 \prec_{i-1} N_2$ (else, $M \simeq N$), we get $\varphi(M_1)\varphi(M_2) < \varphi(N_1)\varphi(N_2)$ by induction hypothesis (notice that the value of $\varphi$ is always positive). By the definition of $\varphi$, it means $\varphi(M)< \varphi(N)$ (to show that $\varphi(M)=\varphi(M_1)\varphi(M_2)$ etc, we use Church-Rosser Theorem again).\\
   \quad This ends the proof. 
 \end{proof}

\section{Existence of Principal types}
 \quad In this section, we will prove (without using well-foundedness established above) that for every valid context $\Gamma$, each $\Gamma$- object has a principal type. Also, the arguments below exhibits how to compute them.
\subsection{Preliminaries}
\begin{lemma}
\label{lemma: 101}
If $A \preceq A'$, then for every variable $x$ and term $N$ we have $[N/x]A \preceq [N/x]A'$.
\qed
\end{lemma}
\begin{lemma}
\label{lemma: 102}
For every term $A$ and universe $U$ with $A \preceq U$, there is a universe $U' \preceq U$ 
such that $A \simeq U'$.
\qed
\end{lemma}
\begin{notation}
We write $Type_{-1}$ for $Prop$ and $Type$ for $Type_j$ for some $j \geq 0$.
\end{notation}
We put * to the references of theorems from \cite{Luo} (i.e., Theorem 1.1* means the theorem 1.1 in \cite{Luo}).
\subsection{Proofs}
\quad To prove the existence of the principal type of every $\Gamma$-object, we restrict the type derivation system of $ECC$
to the system we call $ECC^-$ defined as follows.
\begin{definition}
The rules of $ECC^-$ are those of $ECC$ with $(\Pi2)(\Sigma)(app)(pair)(\preceq)$ replaced by the following: 
\begin{prooftree}
	\AxiomC{$\Gamma \vdash A : Type_j$}
	\AxiomC{$\Gamma, x : A \vdash B : Type_k$}
	\LeftLabel{$(\Pi2')$}
	\RightLabel{$(l = \max\{j, k, 0\} \& k \geq 0)$}
	\BinaryInfC{$\Gamma \vdash \Pi x : A. B : Type_l$}
\end{prooftree}
\begin{prooftree}
	\AxiomC{$\Gamma \vdash A : Type_j$}
	\AxiomC{$\Gamma, x : A \vdash B : Type_k$}
	\LeftLabel{$(\Sigma')$}
	\RightLabel{$(l = \max\{j, k, 0\})$}
	\BinaryInfC{$\Gamma \vdash \Sigma x : A. B : Type_l$}
\end{prooftree}
\begin{prooftree}
	\AxiomC{$\Gamma \vdash M : \Pi x : A. B$}
	\AxiomC{$\Gamma \vdash N : A'$}
	\LeftLabel{$(app')$}
	\RightLabel{$(A' \preceq A)$}
	\BinaryInfC{$\Gamma \vdash MN : [N/x]B$}
\end{prooftree}
\begin{prooftree}
	\AxiomC{$\Gamma \vdash M : A$}
	\AxiomC{$\Gamma \vdash N : C$}
	\AxiomC{$\Gamma, x : A' \vdash B' : Type$}
	\LeftLabel{$(pair')$}
	\RightLabel{$(A \preceq A',\ C \preceq [M/x]B')$}
	\TrinaryInfC{$\Gamma \vdash \langle M, N \rangle_{\Sigma x : A'. B'} : \Sigma x : A'. B'$}
\end{prooftree}
In addition, the following scheme $(\simeq)_{\rho}$ is added:
\begin{prooftree}
	\AxiomC{$\Gamma \vdash M : A$}
	\LeftLabel{$(\simeq)_{\rho}$}
	\RightLabel{$(A \simeq A'$, and $\rho$ is an \underline{$ECC$-derivation} of $\Gamma \vdash A' : Type)$}
	\UnaryInfC{$\Gamma \vdash M : A'$}
\end{prooftree}
(notice that when $(\simeq)_{\rho}$ is applied, an $ECC$-derivation $\rho$ must be specified).
\qed
\end{definition}
It is easy to see that an $ECC^-$-derivation can be naturally simulated by an $ECC$-derivation:
\begin{definition}
For an $ECC^-$-derivation $\theta$, $\mathcal{F}(\theta)$ denotes an $ECC$-derivation constructed inductively as follows:
 \begin{enumerate}
     \item Replace $(\Pi2')$ such as 
     \begin{prooftree}
     \AxiomC{$\vdots\ \delta_1$}
	 \UnaryInfC{$\Gamma \vdash A : Type_j$}
	 \AxiomC{$\vdots\ \delta_2$}
	 \UnaryInfC{$\Gamma, x : A \vdash B : Type_k$}
	 \LeftLabel{$(\Pi2')$}
	 \RightLabel{$(l = \max\{j, k, 0\} \& k \geq 0)$}
	 \BinaryInfC{$\Gamma \vdash \Pi x : A. B : Type_l$}
     \end{prooftree}
     with $(\Pi2)$ as follows:
     \begin{prooftree}
     \AxiomC{$\vdots\ \mathcal{F}(\delta_1)$}
	 \UnaryInfC{$\Gamma \vdash A : Type_j$}
	   \AxiomC{$\vdots\ \delta_3$}
	   \UnaryInfC{$\Gamma \vdash Type_l:Type$}
	  \RightLabel{$(\preceq)$}
	 \BinaryInfC{$\Gamma \vdash A:Type_l$}
	 \AxiomC{$\vdots\ \mathcal{F}(\delta_2)$}
	 \UnaryInfC{$\Gamma, x : A \vdash B : Type_k$}
	  \AxiomC{$\vdots\ \delta_4$}
	  \UnaryInfC{$\Gamma, x:A \vdash Type_l:Type$}
	  \RightLabel{$(\preceq)$}
	 \BinaryInfC{$\Gamma, x:A \vdash B:Type_l$}
	 \RightLabel{$(\Pi2)$}
	 \BinaryInfC{$\Gamma \vdash \Pi x : A. B : Type_l$}
     \end{prooftree}
     ($\delta_3,\delta_4$ above can be obtained by $(Ax)(T)$ and Lemma 3.12* for $ECC$).
     \item Replace $(\Sigma')$ with $(\Sigma)$ in a similar way.
     \item Replace $(app')$ such as
     \begin{prooftree}
      \AxiomC{$\vdots\ \delta_1$}
	  \UnaryInfC{$\Gamma \vdash M : \Pi x : A. B$}
	   \AxiomC{$\vdots\ \delta_2$}
	   \UnaryInfC{$\Gamma \vdash N : A'$}
	  \LeftLabel{$(app')$}
	  \RightLabel{$(A' \preceq A)$}
	  \BinaryInfC{$\Gamma \vdash MN : [N/x]B$}
\end{prooftree}
 with $(app)$ as follows:
 \begin{prooftree}
      \AxiomC{$\vdots\ \mathcal{F}(\delta_1)$}
	  \UnaryInfC{$\Gamma \vdash M : \Pi x : A. B$}
	   \AxiomC{$\vdots\ \mathcal{F}(\delta_2)$}
	   \UnaryInfC{$\Gamma \vdash N : A'$}
	     \AxiomC{$\vdots\ \delta_3$}
	     \UnaryInfC{$\Gamma \vdash A:Type$}
	  \RightLabel{$(\preceq)$}
	  \BinaryInfC{$\Gamma \vdash N:A$}
	\RightLabel{$(app)$}
	\BinaryInfC{$\Gamma \vdash MN : [N/x]B$}
\end{prooftree}
(an appropriate $\delta_3$ exists by the following reason; first, $\mathcal{F}(\delta_1)$ and Theorem 3.15* for $ECC$ yields an $ECC$-derivation of $\Gamma \vdash \Pi x:A.B :U$ where $U$ denotes some universe. Then, it can easily be verified that this derivation has a subderivation which says $\Gamma \vdash A:U'$ where $U'$ denots another universe).
     \item Replace $(pair')$ with $(pair)$ in a similar way.
     \item Replace instances of $(\simeq)_{\rho}$ such as
     \begin{prooftree}
     \AxiomC{$\vdots \delta$}
	 \UnaryInfC{$\Gamma \vdash M : A$}
	  \LeftLabel{$(\simeq)_{\rho}$}
    	\RightLabel{$(A \simeq A'$, and $\rho$ is an $ECC$-derivation of $\Gamma \vdash A' : Type)$}
	  \UnaryInfC{$\Gamma \vdash M : A'$}
     \end{prooftree}
      with
    \begin{prooftree}
     \AxiomC{$\vdots \mathcal{F}(\delta)$}
	 \UnaryInfC{$\Gamma \vdash M : A$}
	  \AxiomC{$\vdots\ \rho$}
	  \UnaryInfC{$\Gamma \vdash A' : Type$}
	  \RightLabel{$(\preceq)$}
	  \BinaryInfC{$\Gamma \vdash M : A'$}
    \end{prooftree}
 \end{enumerate}
 Note that $\mathcal{F}(\theta)$ has the same conclusion as $\theta$.
\end{definition}
Also, it is easy to verify Lemma 3.10* for $ECC^-$ (just apply Lemma 3.10 for $ECC$ to $\mathcal{F}(\theta)$). We omit proofs.\\
\null \quad In $ECC^-$, every $\Gamma$-object has a unique type modulo conversion.
\begin{prop}\label{uni}
In $ECC^-$, $\Gamma \vdash M : A$ and $\Gamma \vdash M : A'$ implies $A \simeq A'$.
\end{prop}
\begin{proof}
We prove by induction of the sum of the heights of the given two derivation trees. 
If one of the derivation ends with
$(\simeq)_{\rho}$, then we can directly apply the induction hypothesis. 
Thus, we may assume that neither of the last rules is not $(\simeq)_{\rho}$. 
The only pairs of rules which can be applied to a term of the same form is $(\Pi1)$ \& $(\Pi2')$. 
However, if they are the last two rules, then there would be a contradiction to the induction hypothesis
($Prop \simeq Type$).
Therefore, we may assume that the last rules of the two derivations are the same.
\par
If the last rules are among $(Ax)(C)(T)(\Pi1)$, then the assertion immediately follows
(note that Chuch-Rosser Theorem implies $C \simeq D \Rightarrow C \equiv D$).
The cases ($\Pi2'$)($\Sigma'$) follow from the induction hypothesis and the definitions of the rules.
For the rules $(\lambda)(\pi1)$ one can use the fact that
$\Pi[\Sigma] x : A. B \simeq \Pi[\Sigma] x : A'. B' \Leftrightarrow A \simeq A' \& B \simeq B'$.
Similarly, for $(app')(\pi2)$ we note that $B \simeq B' \Rightarrow [N/x]B \simeq [N/x]B'$ holds.
As for $(var)$, $A \equiv A'$ follows from Lemma 3.10* for $ECC^-$.
The rest is $(pair')$, which, by definition, always assign the same type $B$ to a term $\langle M,N \rangle_B$.
\end{proof}
The existence of the principal type for every $\Gamma$-object is an immediate consequence of the next proposition:
\begin{prop}\label{transform}
Every derivation tree 
\begin{prooftree}
	\AxiomC{$\vdots$}
	\UnaryInfC{$\Delta \vdash \alpha : \tau$}
\end{prooftree}
of $ECC$ can be transformed into an $ECC^-$-derivation of the following form:
\begin{prooftree}
	\AxiomC{$\vdots\ \mathcal{F}(\theta)$}
	\UnaryInfC{$\Delta \vdash \alpha : \tau'$}
		   \AxiomC{$\vdots\ \eta$}
	       \UnaryInfC{$\Delta \vdash \tau : Type$}
	\RightLabel{$(\preceq)$}
	\BinaryInfC{$\Delta \vdash \alpha : \tau$}
\end{prooftree}
where $\theta$ is an $ECC^-$-derivation deriving $\Delta \vdash \alpha:\tau'$, $\eta$ is an $ECC$-derivation deriving $\Delta \vdash \tau:Type$, and $\tau' \preceq \tau$ (recall that $Type$ denotes some $Type_j$).
\end{prop}
\begin{proof}
The proof is by induction on the height of the given derivation tree of $ECC$.\\ 
\null \quad If the last rule in the derivation is $(Ax)$, it is obvious.\\
\null \quad Suppose the last rule is $(C)$.
The given derivation tree has the following form:
\begin{prooftree}
	\AxiomC{\vdots}
	\UnaryInfC{$\Gamma \vdash A:Type_j$}
		\RightLabel{$(C)$}
	\UnaryInfC{$\Gamma, x:A \vdash Prop:Type_0$}
\end{prooftree}

Applying the induction hypothesis, we get
\begin{prooftree}
	\AxiomC{$\vdots\ \mathcal{F}(\theta)$}
	\UnaryInfC{$\Gamma \vdash A : X$}
		   \AxiomC{$\vdots\ \eta$}
	       \UnaryInfC{$\Gamma \vdash Type_j : Type$}
	\RightLabel{$(\preceq)$}
	\BinaryInfC{$\Gamma \vdash A : Type_j$}
\end{prooftree}
Noting that $X \preceq Type_j$ yields $X \simeq Type_k$ for some $k \leq j$, we get an $ECC^-$-derivation $\theta'$:
\begin{prooftree}
	\AxiomC{$\vdots\ \theta$}
	\UnaryInfC{$\Gamma \vdash A : X$}
	\RightLabel{$(\simeq)_{\rho}$}
	\UnaryInfC{$\Gamma \vdash A : Type_k$}
	\RightLabel{$(C)$}
	\UnaryInfC{$\Gamma \vdash Prop:Type_0$}
\end{prooftree}
(An appropriate $\rho$ can be obtained by $(Ax)(T)$ and Lemma3.12* for $ECC$).
Now, 
\begin{prooftree}
	\AxiomC{$\vdots\ \mathcal{F}(\theta')$}
	\UnaryInfC{$\Gamma \vdash Prop:Type0$}
		   \AxiomC{$\vdots\ \eta'$}
	       \UnaryInfC{$\Gamma \vdash Type_0 : Type$}
	\RightLabel{$(\preceq)$}
	\BinaryInfC{$\Gamma \vdash Prop : Type_0$}
\end{prooftree}
($\eta'$ is obtained by $(Ax)(T)$ and Lemma 3.12* for $ECC$) gives the desired transformation.\\
\null \quad The case $(T)(var)(\Pi 1)$ are dealt with similarly.\\
\null \quad Consider the case $(\Sigma')$.
The given derivation has a form of:
\begin{prooftree}
	\AxiomC{$\vdots$}
	\UnaryInfC{$\Gamma \vdash A : Type_j$}
		   \AxiomC{$\vdots$}
	       \UnaryInfC{$\Gamma, x:A \vdash B : Type_j$}
	\RightLabel{$(\Sigma)$}
	\BinaryInfC{$\Gamma \vdash \Sigma x :A. B Type_j$}
\end{prooftree}
By the induction hypothesis, we get:
\begin{prooftree}
	\AxiomC{$\vdots\ \mathcal{F}(\theta_1)$}
	\UnaryInfC{$\Gamma \vdash A : X_1$}
		   \AxiomC{$\vdots\ \eta_1$}
	       \UnaryInfC{$\Gamma \vdash Type_j : Type$}
	\RightLabel{$(\preceq)$}
	\BinaryInfC{$\Gamma \vdash A : Type_j$}
\end{prooftree}
and
\begin{prooftree}
	\AxiomC{$\vdots\ \mathcal{F}(\theta_2)$}
	\UnaryInfC{$\Gamma, x:A \vdash B : X_2$}
		   \AxiomC{$\vdots\ \eta_2$}
	       \UnaryInfC{$\Gamma, x:A \vdash Type_j : Type$}
	\RightLabel{$(\preceq)$}
	\BinaryInfC{$\Gamma, x:A \vdash B : Type_j$}
\end{prooftree}
As in the previous argument, $X_1 \simeq Type_k$ and $X_2 \simeq Type_l$ for some $k,l \leq j$.
Using these, we get the following $ECC^-$-derivation $\theta'$:
\begin{prooftree}
	\AxiomC{$\vdots\ \theta_1$}
	\UnaryInfC{$\Gamma \vdash A : X_1$}
	\RightLabel{$(\simeq)_{\rho_1}$}
	\UnaryInfC{$\Gamma \vdash A:Type_k$}
		   \AxiomC{$\vdots\ \theta_2$}
	       \UnaryInfC{$\Gamma, x:A \vdash B : X_2$}
	       \RightLabel{$(\simeq)_{\rho_2}$}
           \UnaryInfC{$\Gamma, x:A \vdash B:Type_l$}
	\RightLabel{$(\Sigma')$}
	\BinaryInfC{$\Gamma \vdash \Sigma x:A.B : Type_{\max\{0,k,l\}}$}
\end{prooftree}
(again, $\rho_1, \rho_2$ can be obtained by $(Ax)(T)$ and Lemma 3.12* for $ECC$). Now, the following gives the desired transformation:
\begin{prooftree}
	\AxiomC{$\vdots\ \mathcal{F}(\theta')$}
	\UnaryInfC{$\Gamma \vdash \Sigma x:A.B : Type_{\max\{0,k,l\}}$}
		   \AxiomC{$\vdots\ \eta'$}
	       \UnaryInfC{$\Gamma \vdash Type_j : Type$}
	\RightLabel{$(\preceq)$}
	\BinaryInfC{$\Gamma \vdash \Sigma x:A.B: Type_j$}
\end{prooftree}
($\eta'$ is obtained by Lemma 3.12* for $ECC$).\\
\null \quad The case $(\Pi2')$ can be dealt with similarly.\\
\null \quad As for the case $(\lambda)$, the given derivation in in the form of:
\begin{prooftree}
	\AxiomC{$\vdots$}
	\UnaryInfC{$\Gamma, x:A \vdash M : B$}
	\RightLabel{$(\lambda)$}
	\UnaryInfC{$\Gamma \vdash \lambda x :A. M : \Pi x:A.B$}
\end{prooftree}
By the induction hypothesis, we get:
\begin{prooftree}
	\AxiomC{$\vdots\ \mathcal{F}(\theta)$}
	\UnaryInfC{$\Gamma, x:A \vdash M : B'$}
		   \AxiomC{$\vdots\ \eta$}
	       \UnaryInfC{$\Gamma,x:A \vdash B : Type$}
	\RightLabel{$(\preceq)$}
	\BinaryInfC{$\Gamma \vdash M : B$}
\end{prooftree}
So, we get the following $ECC^-$-derivation $\theta'$:
\begin{prooftree}
	\AxiomC{$\vdots\ \theta$}
	\UnaryInfC{$\Gamma, x:A \vdash M : B'$}
	\UnaryInfC{$\Gamma \vdash \lambda x:A. M : \Pi x:A. B'$}
\end{prooftree}
Now, the following gives the desired transformation:
\begin{prooftree}
	\AxiomC{$\vdots\ \mathcal{F}(\theta')$}
	\UnaryInfC{$\Gamma \vdash \lambda x:A.M : \Pi x:A.B'$}
		   \AxiomC{$\vdots\ \eta'$}
	       \UnaryInfC{$\Gamma \vdash \Pi x:A.B : Type$}
	\RightLabel{$(\preceq)$}
	\BinaryInfC{$\Gamma \vdash \lambda x:A.M: \Pi x:A.B$}
\end{prooftree}
(Note that $\Pi x:A.B' \preceq \Pi x:A.B$ follows from $B' \preceq B$. Also, $\eta'$ can be obtained by applying Theorem 3.15* to the given $ECC$-derivation above).\\
\null \quad Consider the case $(pair)$. The given $ECC$-derivation is as follows:
\begin{prooftree}
	\AxiomC{$\vdots$}
	\UnaryInfC{$\Gamma \vdash M : A$}
		\AxiomC{$\vdots$}
		\UnaryInfC{$\Gamma \vdash N : [M/x]B$}
			\AxiomC{$\vdots$}
			\UnaryInfC{$\Gamma,x:A \vdash B : Type_j$}
	\TrinaryInfC{$\Gamma \vdash \langle M, N \rangle_{\Sigma x : A. B} : \Sigma x: A. B$}
\end{prooftree}
The induction hypothesis yields the trees:
\begin{prooftree}
	\AxiomC{$\vdots\ \mathcal{F}(\theta_1)$}
	\UnaryInfC{$\Gamma \vdash M : A'$}
		   \AxiomC{$\vdots\ \eta_1$}
	       \UnaryInfC{$\Gamma \vdash A : Type$}
	\RightLabel{$(\preceq)$}
	\BinaryInfC{$\Gamma \vdash M : A$}
\end{prooftree}
\begin{prooftree}
	\AxiomC{$\vdots\ \mathcal{F}(\theta_2)$}
	\UnaryInfC{$\Gamma \vdash N : B'$}
		   \AxiomC{$\vdots\ \eta_2$}
	       \UnaryInfC{$\Gamma \vdash [M/x]B : Type$}
	\RightLabel{$(\preceq)$}
	\BinaryInfC{$\Gamma \vdash N : [M/x]B$}
\end{prooftree}
\begin{prooftree}
	\AxiomC{$\vdots\ \mathcal{F}(\theta_3)$}
	\UnaryInfC{$\Gamma, x:A \vdash B : X$}
		   \AxiomC{$\vdots\ \eta_3$}
	       \UnaryInfC{$\Gamma,x:A \vdash Type_j : Type$}
	\RightLabel{$(\preceq)$}
	\BinaryInfC{$\Gamma \vdash B : Type_j$}
\end{prooftree}
(again, $X \simeq Type_k$ for some $k \leq j$). Hence, we get the following $ECC^-$-derivation $\theta'$:
\begin{prooftree}
	\AxiomC{$\vdots\ \theta_1$}
	\UnaryInfC{$\Gamma \vdash M : A'$}
	   \AxiomC{$\vdots\ \theta_2$}
	   \UnaryInfC{$\Gamma \vdash N : B'$}
	     \AxiomC{$\vdots\ \theta_3$}
	     \UnaryInfC{$\Gamma, x:A \vdash B:X$}
	     \RightLabel{$(\simeq)_{\rho}$}
	     \UnaryInfC{$\Gamma,x:A \vdash B:Type_k$}
	\RightLabel{$(pair')$}
	\TrinaryInfC{$\Gamma \vdash \langle M,N \rangle_{\Sigma x:A.B} : \Sigma x:A.B$}
\end{prooftree}
(again, $\rho$ is obtained by $(Ax)(T)$ and Lemma 3.12* for $ECC$). Now, the following gives the desired transformation:
\begin{prooftree}
	\AxiomC{$\vdots\ \mathcal{F}(\theta')$}
	\UnaryInfC{$\Gamma \vdash \langle M,N \rangle_{\Sigma x:A.B} : \Sigma x:A.B$}
		   \AxiomC{$\vdots\ \eta'$}
	       \UnaryInfC{$\Gamma \vdash \Sigma x:A.B : Type$}
	\RightLabel{$(\preceq)$}
	\BinaryInfC{$\langle M,N \rangle_{\Sigma x:A.B} : \Sigma x:A.B$}
\end{prooftree}
(Note that $\eta'$ can be obtained by applying Lemma 3.15* to the given $ECC$-derivation above).\\
\null \quad Consider the case $(app)$.
\begin{prooftree}
	\AxiomC{$\vdots$}
	\UnaryInfC{$\Gamma \vdash M : \Pi x:A.B$}
		   \AxiomC{$\vdots$}
	       \UnaryInfC{$\Gamma \vdash N:A$}
	\RightLabel{$(app)$}
	\BinaryInfC{$\Gamma \vdash MN: [N/x]B$}
\end{prooftree}
By the induction hypothesis, we get:
\begin{prooftree}
	\AxiomC{$\vdots\ \mathcal{F}(\theta_1)$}
	\UnaryInfC{$\Gamma \vdash M : X_1$}
		   \AxiomC{$\vdots\ \eta_1$}
	       \UnaryInfC{$\Gamma \vdash \Pi x:A.B :Type$}
	\RightLabel{$(\preceq)$}
	\BinaryInfC{$\Gamma \vdash M : \Pi x:A.B$}
\end{prooftree}
and
\begin{prooftree}
	\AxiomC{$\vdots\ \mathcal{F}(\theta_2)$}
	\UnaryInfC{$\Gamma \vdash N : X_2$}
		   \AxiomC{$\vdots\ \eta_2$}
	       \UnaryInfC{$\Gamma \vdash A : Type$}
	\RightLabel{$(\preceq)$}
	\BinaryInfC{$\Gamma \vdash N : A$}
\end{prooftree}
(since $X_1 \preceq \Pi x:A.B$, $X_1 \simeq \Pi x:A.B_1$ for some $B_1 \preceq B$. Using Church-Rosser Theorem, we may assume that $X_1 \triangleright \Pi x:A.B_1$). 
Since $\mathcal{F}(\theta_1)$ derives $\Gamma \vdash M:X_1$, $X_1$ is a $\Gamma$-type by Lemma 3.15 for $ECC$. Hence, $\Pi x:A.B_1$ is also a $\Gamma$-type because of Theorem 3.16* for $ECC$.
Now, we can get the following $ECC^-$-derivation $\theta'$:
\begin{prooftree}
	\AxiomC{$\vdots\ \theta_1$}
	\UnaryInfC{$\Gamma \vdash M:X_1$}
	\RightLabel{$(\simeq)_{\rho}$}
	\UnaryInfC{$\Gamma \vdash M: \Pi x:A.B_1$}
		   \AxiomC{$\vdots\ \theta_2$}
	       \UnaryInfC{$\Gamma \vdash N:X_2$}
	\RightLabel{$(app')$}
	\BinaryInfC{$\Gamma \vdash MN: [N/x]B_1$}
\end{prooftree}
(An appropriate $\rho$ exists since $\Pi x:A. B_1$ is a $\Gamma$-type in $ECC$).
Then, the following derivation gives the desired transformation:
\begin{prooftree}
	\AxiomC{$\vdots\ \mathcal{F}(\theta')$}
	\UnaryInfC{$\Gamma \vdash MN : [N/x]B_1$}
		   \AxiomC{$\vdots\ \eta'$}
	       \UnaryInfC{$\Gamma \vdash [N/x]B : Type$}
	\RightLabel{$(\preceq)$}
	\BinaryInfC{$\Gamma \vdash MN : [N/x]B$}
\end{prooftree}
($\eta'$ is obtained by applying Theorem 3.15* to the given $ECC$-derivation. Note that $[N/x]B_1 \preceq [N/x]B$).\\
\null \quad The cases $(\pi 1)(\pi 2)$ can be dealt with in similar ways. \\
\null \quad As for the case $(\preceq)$, the given derivation is as follows:
\begin{prooftree}
	\AxiomC{$\vdots$}
	\UnaryInfC{$\Gamma \vdash M : A$}
		   \AxiomC{$\vdots\ \eta'$}
	       \UnaryInfC{$\Gamma \vdash B:Type$}
	\RightLabel{$(\preceq)$}
	\BinaryInfC{$\Gamma \vdash M: B$}
\end{prooftree}
By the induction hypothesis, we get:
\begin{prooftree}
	\AxiomC{$\vdots\ \mathcal{F}(\theta)$}
	\UnaryInfC{$\Gamma \vdash M : A'$}
		   \AxiomC{$\vdots\ \eta$}
	       \UnaryInfC{$\Gamma \vdash A :Type$}
	\RightLabel{$(\preceq)$}
	\BinaryInfC{$\Gamma \vdash M : A$}
\end{prooftree}
Since $A' \preceq A \preceq B$, we get
\begin{prooftree}
	\AxiomC{$\vdots\ \mathcal{F}(\theta)$}
	\UnaryInfC{$\Gamma \vdash M : A'$}
		   \AxiomC{$\vdots\ \eta'$}
	       \UnaryInfC{$\Gamma \vdash B :Type$}
	\RightLabel{$(\preceq)$}
	\BinaryInfC{$\Gamma \vdash M : B$}
\end{prooftree}
This completes the proof.
\end{proof}

\begin{theorem}
In $ECC$, every $\Gamma$-object has a principal type.
\end{theorem}

\begin{proof}
 Assume that $\Gamma \vdash \alpha:\tau$ can be derived in $ECC$. Then, by Proposition \ref{transform}, we can obtain a derivation in the form of:
 \begin{prooftree}
	\AxiomC{$\vdots\ \mathcal{F}(\theta)$}
	\UnaryInfC{$\Gamma \vdash \alpha : \tau'$}
		   \AxiomC{$\vdots\ \eta$}
	       \UnaryInfC{$\Gamma \vdash \tau : Type$}
	\RightLabel{$(\preceq)$}
	\BinaryInfC{$\Gamma \vdash \alpha : \tau$}
\end{prooftree}
(where $\theta$ is an $ECC^-$ derivation deriving $\Gamma \vdash \alpha:\tau'$, $\eta$ is an $ECC$-derivation deriving $\Gamma \vdash \tau:Type$, and $\tau' \preceq \tau$).\\
\null \quad Let us show that this $\tau'$ is the principal type of $\alpha$ in $\Gamma$. Indeed, if $\Gamma \vdash \alpha : \sigma$ can be derived in $ECC$, then Proposition \ref{transform} yields an $ECC$-derivation of the following form:
 \begin{prooftree}
	\AxiomC{$\vdots\ \mathcal{F}(\theta_2)$}
	\UnaryInfC{$\Gamma \vdash \alpha : \sigma'$}
		   \AxiomC{$\vdots $}
	       \UnaryInfC{$\Gamma \vdash \sigma : Type$}
	\RightLabel{$(\preceq)$}
	\BinaryInfC{$\Gamma \vdash \alpha : \sigma$}
\end{prooftree}
(where $\theta_2$ is an $ECC^-$-derivation).\\
\null \quad Then, using Proposition \ref{uni}, we get $\tau' \simeq \sigma' \preceq \sigma$.
 
\end{proof}


\begin{thebibliography}{9}
\bibitem{Luo}
Z. Luo, Computation and Reasoning: A Type Theory for ComputerScience. New York, NY, USA: Oxford University Press, Inc, 1994.

\end{thebibliography}
\end{document}